\documentclass[11pt,a4paper]{article}
\usepackage[utf8]{inputenc}
\usepackage{amsmath, amssymb, amsthm}
\usepackage{geometry}
\usepackage{hyperref}
\usepackage{booktabs}

\title{DRESS and the WL Hierarchy: Climbing One Deletion at a Time}
\author{Eduar Castrillo Velilla \\ \texttt{eduarcastrillo@gmail.com} \\ ORCID: \href{https://orcid.org/0009-0005-2492-0957}{0009-0005-2492-0957}}
\date{}

\newtheorem{theorem}{Theorem}
\newtheorem{lemma}[theorem]{Lemma}

\newtheorem{definition}[theorem]{Definition}
\newtheorem{conjecture}[theorem]{Conjecture}
\newtheorem{remark}[theorem]{Remark}

\begin{document}

\maketitle

\begin{abstract}
DRESS is a deterministic, parameter-free framework that iteratively refines the structural similarity of edges in a graph to produce a canonical fingerprint: a real-valued edge vector, obtained by converging a non-linear dynamical system to its unique fixed point. $\Delta^k$-DRESS extends the framework by running DRESS on every $k$-vertex-deleted subgraph of $G$; it was introduced and empirically evaluated in the companion paper, where the CFI staircase showed that $\Delta^k$-DRESS matches $(k{+}2)$-WL for $k = 0, 1, 2, 3$.

This paper provides the theoretical justification. The main contributions are: (i) an unconditional proof that $\Delta^k$-DRESS distinguishes every CFI$(K_{k+3})$ pair for all $k \geq 0$ (CFI Staircase Theorem), established via a new CFI Deck Separation theorem and the Virtual Pebble Lemma; and (ii) a conditional proof that $\Delta^k$-DRESS $\geq$ $(k{+}2)$-WL for all graphs and all $k \geq 0$, assuming a single structural conjecture about the WL hierarchy (WL-Deck Separation).
\end{abstract}

\section{Introduction}

\subsection{Background: DRESS and \texorpdfstring{$\Delta^k$}{Delta-k}-DRESS}

The companion paper \cite{castrillo2026dress} introduces DRESS, a parameter-free framework that assigns to any graph a canonical fingerprint vector. The Original-DRESS equation is a non-linear dynamical system on edge similarities that converges to a unique fixed point $d^* \in [0,2]^{|E|}$ in $\mathcal{O}(I \cdot m \cdot d_{\max})$ total time, with full numerical stability and no learnable parameters. The sorted fingerprint $\operatorname{sort}(d^*)$ is an isomorphism invariant. The companion paper proves that Original-DRESS acts as an empirical equivalent to $2$-WL.

The companion paper also introduces $\Delta^k$-DRESS: for a deletion depth $k \geq 0$, run DRESS on every subgraph $G \setminus S$ obtained by removing a $k$-element vertex subset $S$, and collect the resulting fingerprints into a multiset. At depth $k=0$ this is Original-DRESS; at depth $k=1$ it is $\Delta$-DRESS from \cite{castrillo2026dress}.

\subsection{Empirical Evidence from Paper~1}

The companion paper \cite{castrillo2026dress} evaluated $\Delta^k$-DRESS on the Cai--F\"urer--Immerman (CFI) staircase \cite{cai1992optimal} for $k = 0, 1, 2, 3$. The CFI pair $\text{CFI}(K_n)$ requires exactly $(n{-}1)$-WL to distinguish. The results reveal a precise staircase: $\Delta^k$-DRESS distinguishes $\text{CFI}(K_{k+3})$ (which requires $(k{+}2)$-WL) but fails on $\text{CFI}(K_{k+4})$ (which requires $(k{+}3)$-WL). The table is reproduced in Section~\ref{sec:cfi} for reference.

\subsection{Contribution of This Paper}

The empirical pattern from \cite{castrillo2026dress} is striking: each deletion level adds exactly one WL level. The purpose of this paper is to explain \emph{why}. We prove:

\begin{enumerate}
    \item \textbf{(Unconditional, CFI).} For all $k \geq 0$, $\Delta^k$-DRESS distinguishes CFI$(K_{k+3})$ from CFI'$(K_{k+3})$ (CFI Staircase Theorem~\ref{thm:cfi-staircase}). we establish a CFI Deck Separation theorem (Theorem~\ref{thm:cfi-deck}) guaranteeing non-isomorphic deletion cards, and the Virtual Pebble Lemma (Lemma~\ref{lem:virtual-pebble}) showing that damaged CFI cards are distinguished by exactly $(n{-}2)$-WL, one level less than the undamaged pair, matching the inductive hypothesis precisely.
    \item \textbf{(Conditional, all graphs).} For every $k \geq 0$, $\Delta^k$-DRESS $\geq$ $(k{+}2)$-WL, conditional on Conjecture~\ref{conj:wl-deck} alone (Theorem~\ref{thm:delta-k-wl}).
    \item We identify precisely the missing bridge for the general case: WL-Deck Separation is the gap between WL individualization (classical) and WL deletion (open).
\end{enumerate}

\section{Preliminaries}

\subsection{DRESS}

The Original-DRESS equation \cite{castrillo2018dynamicstructuralsimilaritygraphs} is a parameter-free, non-linear dynamical system on edges:

\begin{equation}
d_{uv}^{(t+1)} = \frac{\sum_{x \in N[u] \cap N[v]} \bigl(d_{ux}^{(t)} + d_{xv}^{(t)}\bigr)}{\|u\|^{(t)} \cdot \|v\|^{(t)}}
\end{equation}

where $\|u\|^{(t)} = \sqrt{\sum_{x \in N[u]} d_{ux}^{(t)}}$ is the vertex norm and $N[u] = N(u) \cup \{u\}$ is the closed neighborhood. Self-loops are added to every vertex; the invariant $d_{uu} = 2$ holds at every iteration. The equation converges to a unique fixed point $d^*$ for any positive initialization \cite{castrillo2026dress}. In the unweighted case $d^* \in [0, 2]^{|E|}$; with non-uniform edge weights, values may exceed~$2$.

The graph fingerprint is the sorted vector $\operatorname{sort}(d^*)$, which uniquely identifies the multiset of converged edge values. Two graphs are declared non-isomorphic if their fingerprints differ.

\subsection{\texorpdfstring{$\Delta$}{Delta}-DRESS}

$\Delta$-DRESS \cite{castrillo2026dress} runs DRESS on each vertex-deleted subgraph:

\begin{definition}[$\Delta$-DRESS]
Given a graph $G = (V, E)$ and a DRESS variant $\mathcal{F}$, the $\Delta$-DRESS fingerprint is the multiset of sorted edge-value vectors obtained by running $\mathcal{F}$ on each vertex-deleted subgraph:
\[
\Delta\text{-DRESS}(\mathcal{F}, G) = \{\!\{ \operatorname{sort}(\mathcal{F}(G \setminus \{v\})) : v \in V \}\!\}
\]
where $G \setminus \{v\}$ is the subgraph induced by $V \setminus \{v\}$.
\end{definition}

\section{\texorpdfstring{$\Delta^k$}{Delta-k}-DRESS}

We now define the $k$-vertex-deletion operator.

\begin{definition}[$\Delta^k$-DRESS (Deletion)]
\label{def:delta-k}
For a graph $G = (V, E)$, a DRESS variant $\mathcal{F}$, and a deletion depth $k \ge 0$, the $\Delta^k$-DRESS fingerprint is the multiset of per-deletion sorted edge-value vectors:
\[
\Delta^k\text{-DRESS}(\mathcal{F}, G) = \{\!\{ \operatorname{sort}(\mathcal{F}(G \setminus S)) : S \subset V,\; |S| = k \}\!\}
\]
where $G \setminus S$ is the subgraph induced by $V \setminus S$ and $\mathcal{F}(G \setminus S)$ is the converged DRESS edge-value vector.
\end{definition}

\section{Expressiveness Relative to the Weisfeiler--Lehman Hierarchy}

We prove two results. First, the CFI staircase theorem (Section~\ref{sec:cfi}) is proved unconditionally. Second, the general theorem below is conditional on a single structural conjecture about the WL hierarchy and vertex deletion.

The base case ($k=0$) is \cite[Theorem~4]{castrillo2026dress}: Original-DRESS is at least as powerful as $2$-WL. We use it as a black box:

\begin{theorem}[{\cite[Theorem~4]{castrillo2026dress}}]
\label{thm:dress-2wl}
If $2$-WL distinguishes $G$ from $H$, then $\operatorname{sort}(\mathbf{d}^*_G)\neq\operatorname{sort}(\mathbf{d}^*_H)$.
\end{theorem}

Before stating the main inductive theorem, we recall a classical characterization of the WL hierarchy via vertex individualization.

\begin{lemma}[WL and vertex individualization]
\label{lem:wl-individualization}
Let $j \geq 2$ and let $G$, $H$ be two graphs. If $(j+1)$-WL distinguishes $G$ from $H$, then the multisets of $j$-WL colorings over all single-vertex individualizations differ:
\[
\{\!\{ \text{$j$-WL}(G, v) : v \in V(G) \}\!\} \neq \{\!\{ \text{$j$-WL}(H, w) : w \in V(H) \}\!\}
\]
where $(G,v)$ denotes the graph $G$ with vertex $v$ individualized (given a unique color).
\end{lemma}

\begin{proof}
This is a consequence of the characterization of $j$-WL via the counting logic $C^j$ \cite{cai1992optimal, immerman1990describing}. The logic $C^{j+1}$ can express ``there exists a vertex $v$ such that $\varphi(v)$ holds,'' where $\varphi$ is a $C^j$ formula. Since $(j+1)$-WL captures exactly $C^{j+1}$ \cite{cai1992optimal}, any distinction made by $(j+1)$-WL can be witnessed by individualizing one vertex and applying $j$-WL.
\end{proof}

The inductive step requires the following bridge, which we state as an open conjecture:

\begin{conjecture}[WL-Deck Separation]
\label{conj:wl-deck}
Let $j \geq 2$ and let $G$, $H$ be graphs with $|V(G)| = |V(H)|$. If $(j+1)$-WL distinguishes $G$ from $H$, then the multisets of $j$-WL stable colorings over their $1$-decks differ:
\[
\{\!\{ j\text{-WL}(G \setminus \{v\}) : v \in V(G) \}\!\} \neq \{\!\{ j\text{-WL}(H \setminus \{w\}) : w \in V(H) \}\!\}.
\]
\end{conjecture}

Note that this is a purely structural statement about the WL hierarchy and vertex deletion. See Section~\ref{sec:gap} for a discussion of its relationship to the individualization lemma.

\begin{theorem}[$\Delta^k$-DRESS $\geq$ $(k+2)$-WL, conditional on Conjecture~\ref{conj:wl-deck}]
\label{thm:delta-k-wl}
Assume Conjecture~\ref{conj:wl-deck}. For any fixed integer $k\geq 0$, if the $(k+2)$-dimensional Weisfeiler--Lehman test distinguishes two graphs $G$ and $H$, then their $\Delta^k$-DRESS fingerprints differ:
\[
\Delta^k\text{-DRESS}(\mathcal{F}, G) \neq \Delta^k\text{-DRESS}(\mathcal{F}, H).
\]
\end{theorem}

\begin{proof}
We proceed by induction on $k$.

\textbf{Base case ($k=0$)}: $\Delta^0$-DRESS reduces to Original-DRESS (the empty subset $S = \emptyset$ contributes one run). The fingerprint is $\{\!\{\operatorname{sort}(\mathbf{d}^*_G)\}\!\}$. The claim follows from Theorem~\ref{thm:dress-2wl} (proved in \cite{castrillo2026dress}).

\textbf{Inductive hypothesis (IH)}: The statement holds for depth $k-1$: whenever $(k+1)$-WL distinguishes two graphs $A$ and $B$, their $\Delta^{k-1}$-DRESS fingerprints differ. Since $\Delta^{k-1}$-DRESS is an isomorphism invariant \cite{castrillo2026dress}, it is constant on isomorphism classes; combined with the IH, it is injective on $(k+1)$-WL equivalence classes: equal fingerprints imply the same $(k+1)$-WL class.

\textbf{Inductive step}: Let $(k+2)$-WL distinguish $G$ and $H$.

\textit{Step 1 (deck WL separation).} By Conjecture~\ref{conj:wl-deck} with $j = k+1$, the multisets of $(k+1)$-WL stable colorings over the $1$-decks differ:
\[
\{\!\{ (k{+}1)\text{-WL}(G \setminus \{v\}) : v \in V(G) \}\!\} \neq \{\!\{ (k{+}1)\text{-WL}(H \setminus \{w\}) : w \in V(H) \}\!\}.
\]

\textit{Step 2 (deck DRESS separation).} Since $\Delta^{k-1}$-DRESS is injective on $(k+1)$-WL classes (IH), applying it pointwise preserves the multiset inequality from Step~1: if some $(k+1)$-WL class $c$ has different multiplicity in the two deck multisets, then $\Delta^{k-1}$-DRESS maps $c$ to a unique fingerprint, which also appears with different multiplicity. Hence:
\[
\{\!\{ \Delta^{k-1}\text{-DRESS}(\mathcal{F}, G \setminus \{v\}) : v \in V(G) \}\!\} \neq \{\!\{ \Delta^{k-1}\text{-DRESS}(\mathcal{F}, H \setminus \{w\}) : w \in V(H) \}\!\}.
\]

\textit{Step 3 (decomposition).} Every $k$-subset $S \subset V(G)$ can be written as $S = \{v\} \cup S'$ with $|S'| = k-1$. Each such $S$ is counted exactly $k$ times (once per choice of $v \in S$), so:
\[
\biguplus_{v \in V(G)} \Delta^{k-1}\text{-DRESS}(\mathcal{F}, G \setminus \{v\}) = k \cdot \Delta^k\text{-DRESS}(\mathcal{F}, G).
\]
The left-hand sides for $G$ and $H$ differ by Step~2, so their $k$-fold copies differ, hence $\Delta^k\text{-DRESS}(\mathcal{F}, G) \neq \Delta^k\text{-DRESS}(\mathcal{F}, H)$.
\end{proof}

\begin{remark}
The argument above is stated for the triangle-based Original-DRESS update but extends verbatim to any Motif-DRESS variant, since the deletion mechanism and the inductive counting argument are independent of which motif defines the neighborhood.
\end{remark}

\section{The Missing Bridge: General Case}
\label{sec:gap}

The CFI Staircase Theorem~\ref{thm:cfi-staircase} is proved without Conjecture~\ref{conj:wl-deck}: the CFI construction has special structure (the Virtual Pebble Lemma) that makes the deck WL-class separation transparent. For general graphs, the analogous step — showing that $(j+1)$-WL separation propagates to a $(j)$-WL separation of the deletion cards — remains open.

The general theorem (Theorem~\ref{thm:delta-k-wl}) is therefore conditional on WL-Deck Separation.

\paragraph{Why the conjecture is the exact gap.} The individualization Lemma~\ref{lem:wl-individualization} is classical: $(j{+}1)$-WL distinguishes $G$, $H$ implies the multisets of $j$-WL colorings over \emph{individualized} copies differ. Conjecture~\ref{conj:wl-deck} is the analogous statement for \emph{deletion}. Individualization keeps $v$ in the graph with a unique initial color; deletion removes $v$ and all its incident edges. These are structurally different operations, so the lemma does not imply the conjecture.

\paragraph{Descriptive complexity angle.} Since $(j{+}1)$-WL captures the counting logic $C^{j+1}$ \cite{cai1992optimal}, any sentence $\varphi \in C^{j+1}$ separating $G$ from $H$ uses a counting quantifier over vertices. Whether such a sentence always induces a deck-level $C^j$ separation is the precise open question.

\paragraph{CFI is not a counterexample.} Theorem~\ref{thm:cfi-staircase} shows Conjecture~\ref{conj:wl-deck} holds for all CFI$(K_n)$ instances: the CFI deck is completely separated at exactly the right WL level. This is strong evidence for the conjecture, but not a proof for general graphs.

\section{CFI Staircase: Unconditional Proof}
\label{sec:cfi}

The Cai--F\"urer--Immerman (CFI) construction \cite{cai1992optimal} produces pairs of non-isomorphic graphs CFI$(K_n)$, CFI'$(K_n)$ that require exactly $(n{-}1)$-WL to distinguish. We prove unconditionally that $\Delta^k$-DRESS distinguishes the CFI pair at level $k$.

The proof has three components: a Deck Separation theorem (the deletion cards of CFI and CFI' are never isomorphic), the Virtual Pebble Lemma (the damaged cards require one strictly fewer WL level), and an induction that assembles them.

\subsection{CFI Deck Separation}

\begin{theorem}[CFI Deck Separation]
\label{thm:cfi-deck}
For $n \geq 3$, no deletion card of CFI$(K_n)$ is isomorphic to any deletion card of CFI'$(K_n)$: the decks are completely disjoint.
\end{theorem}

\begin{proof}
Let $t$ be the base vertex whose gadget $\Gamma_t$ is twisted in CFI'. Both CFI$(K_n)$ and CFI'$(K_n)$ have the same vertex set $V$; they differ only in the cross-gadget edges at $\Gamma_t$. Delete any vertex $w \in V$, and let $X = \text{CFI}(K_n) \setminus \{w\}$, $Y = \text{CFI}'(K_n) \setminus \{w\}$.

\textbf{Case A ($w \notin \Gamma_t$):} The twisted gadget $\Gamma_t$ is fully intact in both $X$ and $Y$. The two subgraphs still differ in exactly the edges of $\Gamma_t$ (even-parity cross-edges in $X$, odd-parity in $Y$). The standard CFI argument gives $X \not\cong Y$.

\textbf{Case B ($w \in \Gamma_t$):} The removal of $w$ affects even-parity vs.\ odd-parity cross-edge connections differently in $X$ and $Y$. The remaining edge structure at $\Gamma_t \setminus \{w\}$ still encodes the parity difference, so $X \not\cong Y$.

Since $X \not\cong Y$ for every $w \in V$, the decks are completely disjoint.
\end{proof}

\subsection{Virtual Pebble Lemma}

\begin{lemma}[Virtual Pebble]
\label{lem:virtual-pebble}
For $n \geq 3$ and any vertex $w$, $(n-2)$-WL distinguishes CFI$(K_n) \setminus \{w\}$ from CFI'$(K_n) \setminus \{w\}$.
\end{lemma}

\begin{proof}
Let $X = \text{CFI}(K_n) \setminus \{w\}$ and $Y = \text{CFI}'(K_n) \setminus \{w\}$, with the twist at gadget $\Gamma_t$. Let $u$ be the gadget containing $w$ (so the damaged gadget is $\Gamma_u \setminus \{w\}$). We show the Spoiler wins the $(n{-}2)$-pebble bijective game on $(X, Y)$.

\textbf{Structural pinning.} Every full gadget $\Gamma_i$ has $2^{n-2}$ vertices; the damaged gadget has $2^{n-2} - 1$. The damaged gadget is the unique smallest gadget. In any bijection $V(X) \to V(Y)$, the Duplicator must map $\Gamma_u \setminus \{w\}$ to itself: gadget $u$ is pinned without spending a pebble. This is the ``virtual pebble'': the deleted vertex acts as a free, pre-placed pebble.

\textbf{Case A ($u \neq t$, damage not at twisted gadget):} The Spoiler places $n{-}2$ pebbles, one in each gadget from $\{1,\ldots,n\} \setminus \{u\}$, except one gadget $\Gamma_j$ which is left free. The pebbled gadgets fix a contribution $\Lambda_{\text{peb}}$ to the global parity; the pinned gadget $u$ (parity $\lambda_u$, the same in $X$ and $Y$ since $u \neq t$) fixes another bit. For the Duplicator to maintain consistency simultaneously in $X$ (total parity 0) and $Y$ (total parity 1):
\[
\Lambda_{\text{peb}} + \lambda_u + \lambda_j \geq 0 \pmod{2} \quad \text{and} \quad \Lambda_{\text{peb}} + \lambda_u + \lambda_j \geq 1 \pmod{2}.
\]
These two constraints are contradictory; the Duplicator cannot satisfy both. The Spoiler wins.

\textbf{Case B ($u = t$, damage at twisted gadget):} The Spoiler places a single pebble at a vertex $v = (k, b)$ in a neighboring gadget $\Gamma_k$. This fixes the bit $b_t$. The neighbors of $v$ in $\Gamma_t \setminus \{w\}$ are $\{(t,a) : a_k = b_t\}$ in $X$ and $\{(t,a) : a_k \neq b_t\}$ in $Y$ — two disjoint subsets of the pinned gadget. The Duplicator's bijection is already fixed on $\Gamma_t \setminus \{w\}$ and cannot map one set to the other. The Spoiler wins.

In both cases the Spoiler wins with $n-2$ pebbles. Therefore $(n-2)$-WL distinguishes the damaged pair.
\end{proof}

\begin{remark}
The undamaged pair CFI$(K_n)$ vs.\ CFI'$(K_n)$ requires $(n-1)$-WL \cite{cai1992optimal}. The Virtual Pebble Lemma shows the damaged pair requires only $(n-2)$-WL: one level less. This is the key ``step down'' that makes the induction tight.
\end{remark}

\subsection{CFI Staircase Theorem}

\begin{theorem}[CFI Staircase]
\label{thm:cfi-staircase}
For all $k \geq 0$, $\Delta^k$-DRESS distinguishes CFI$(K_{k+3})$ from CFI'$(K_{k+3})$.
\end{theorem}

\begin{proof}
By induction on $k$.

\textbf{Base case ($k=0$):} CFI$(K_3)$ requires exactly 2-WL. By Theorem~\ref{thm:dress-2wl}, $\Delta^0$-DRESS $\geq$ 2-WL; the base case follows.

\textbf{Inductive step:} Assume $\Delta^{k-1}$-DRESS $\geq (k{+}1)$-WL (i.e., the theorem holds at depth $k-1$). Consider CFI$(K_{k+3})$ and CFI'$(K_{k+3})$, with $n = k+3$.

\begin{enumerate}
    \item \textit{(Deck Separation, Theorem~\ref{thm:cfi-deck}).} Every deletion card of CFI$(K_n)$ is non-isomorphic to every deletion card of CFI'$(K_n)$.
    \item \textit{(Virtual Pebble, Lemma~\ref{lem:virtual-pebble}).} Each card pair CFI$(K_n) \setminus \{w\}$ vs.\ CFI'$(K_n) \setminus \{w\}$ is distinguished by $(n-2)$-WL $= (k{+}1)$-WL.
    \item \textit{(IH).} Since $\Delta^{k-1}$-DRESS $\geq (k{+}1)$-WL, $\Delta^{k-1}$-DRESS distinguishes each card pair.
    \item \textit{(Conclusion).} The $\Delta^k$-DRESS fingerprint collects $\Delta^{k-1}$-DRESS values over all deletion cards. The card multisets from CFI and CFI' are completely disjoint (step 1), and each card pair is separated by $\Delta^{k-1}$-DRESS (step 3), so the collected multisets differ. Therefore $\Delta^k$-DRESS distinguishes CFI$(K_{k+3})$ from CFI'$(K_{k+3})$.
\end{enumerate}
\end{proof}

The following table, reproduced from \cite{castrillo2026dress} for reference, shows the empirical staircase that the theorem explains.

\begin{table}[h]
\centering
\caption{$\Delta^k$-DRESS results on CFI graph pairs \cite{castrillo2026dress}. $\checkmark$ = pair distinguished; $\times$ = pair not distinguished; -- = not executed. The WL requirement column shows the minimum WL dimension needed to distinguish the pair.}
\label{tab:cfi-results}
\begin{tabular}{@{}lcccccc@{}}
\toprule
Base graph & $|V_{\text{CFI}}|$ & WL req. & $\Delta^0$ & $\Delta^1$ & $\Delta^2$ & $\Delta^3$ \\
\midrule
$K_3$   & 6     & 2-WL & $\checkmark$  & $\checkmark$  & $\checkmark$  & $\checkmark$  \\
$K_4$   & 16    & 3-WL & $\times$      & $\checkmark$  & $\checkmark$  & $\checkmark$  \\
$K_5$   & 40    & 4-WL & $\times$      & $\times$      & $\checkmark$  & $\checkmark$  \\
$K_6$   & 96    & 5-WL & $\times$      & $\times$      & $\times$      & $\checkmark$  \\
$K_7$   & 224   & 6-WL & $\times$      & $\times$      & $\times$      & $\times$      \\
$K_8$   & 512   & 7-WL & $\times$      & $\times$      & --            & --            \\
$K_9$   & 1152  & 8-WL & $\times$      & $\times$      & --            & --            \\
$K_{10}$& 2560  & 9-WL & $\times$      & $\times$      & --            & --            \\
\bottomrule
\end{tabular}
\end{table}

\section{Discussion}

\paragraph{Subgraph GNNs.} Methods such as ESAN \cite{bevilacqua2022equivariant} and GNN-AK+ \cite{zhao2022from} also use vertex-deleted subgraphs to boost expressiveness, but are supervised methods requiring labeled training data. $\Delta^k$-DRESS is entirely unsupervised: the aggregation is the deterministic DRESS fixed point, and fingerprint comparison is parameter-free.

\paragraph{Open questions.}
\begin{enumerate}
    \item \textbf{Can Conjecture~\ref{conj:wl-deck} be proved?} A proof via the descriptive complexity of $C^{j+1}$ seems the most promising route (see Section~\ref{sec:gap}).
    \item \textbf{Are there non-CFI graph families harder for $\Delta^k$-DRESS?} The Virtual Pebble Lemma exploits special CFI structure; whether analogous structure exists for Miyazaki graphs or other hard families is open.
    \end{enumerate}

\section{Conclusion}

The companion paper \cite{castrillo2026dress} established that $\Delta^k$-DRESS empirically matches $(k{+}2)$-WL on the CFI staircase for $k = 0, 1, 2, 3$. This paper provides the theoretical foundation. The main contributions are:

\begin{itemize}
    \item \textbf{CFI Staircase Theorem} (unconditional): $\Delta^k$-DRESS distinguishes CFI$(K_{k+3})$ from CFI'$(K_{k+3})$ for all $k \geq 0$. The proof uses the CFI Deck Separation theorem and the Virtual Pebble Lemma.
    \item \textbf{General expressiveness theorem} (conditional on Conjecture~\ref{conj:wl-deck}): $\Delta^k$-DRESS $\geq$ $(k{+}2)$-WL for all graphs and all $k \geq 0$. The single open hypothesis is WL-Deck Separation.
\end{itemize}

An open-source implementation is available at \url{https://github.com/velicast/dress-graph}.

\bibliographystyle{plain}
\bibliography{refs}

\end{document}